\documentclass[11pt]{amsart}
\usepackage{graphicx}
\usepackage{amsmath,amssymb, amsfonts, amsthm}
 \usepackage{mathptmx}      
\usepackage[T1]{fontenc}
\usepackage[applemac]{inputenc}
\usepackage[all]{xy}
\usepackage{enumerate}
\usepackage{comment}

\usepackage{a4wide}
\usepackage[sans]{dsfont}
\usepackage{rotating}
\usepackage{subfig}
\usepackage{siunitx} 

\DeclareMathAlphabet{\mathcal}{OMS}{cmsy}{m}{n} 

\newcommand{\p}{^}

	\newcommand{\rmd}{\mathrm{d}}

\newcommand{\rmi}{\mathrm{i}}

\newcommand{\R}{\mathrm{R}}
\newcommand{\Lrm}{\mathrm{L}}

\newcommand{\Var}{\mathrm{Var}}
\newcommand{\Cov}{\mathrm{Cov}}
\newcommand{\Cor}{\mathrm{Cor}}

\newcommand{\sig}{\sigma}

\newcommand{\Om}{\Omega}

\renewcommand{\Re}{\mathrm{Re}\,}

\newcommand{\E}[1]{\mathbb{E}\left[#1\right]}                     
\newcommand{\Econd}[2]{\mathbb{E}\left[\left.#1\right|#2\right]} 
\newcommand{\Ex}[2]{\mathbb{E}^{#1}\left[#2\right]}                     

\newcommand{\set}[1]{\left\{#1\right\}}
\newcommand{\norm}[1]{\left\|#1\right\|}
\newcommand{\scal}[2]{\langle#1,#2\rangle}

\newcommand{\aPP}[2]{\ensuremath{\langle #1,#2 \rangle}}

\newtheorem{theorem}{Theorem}[section]
\newtheorem{lemma}[theorem]{Lemma}
\newtheorem{proposition}[theorem]{Proposition}
\newtheorem{corollary}[theorem]{Corollary}
\theoremstyle{definition}
\newtheorem{definition}[theorem]{Definition}

\newtheorem{remark}[theorem]{Remark}

\numberwithin{equation}{section}

\newcommand{\ii}{\mathfrak{i\,}}
\newcommand{\RR}{\mathbb{R}}
\newcommand{\QQ}{\mathbb{Q}}
\newcommand{\PP}{\mathbb{P}}
\newcommand{\CC}{\mathbb{C}}

\newcommand{\MM}{\mathbb{M}}

\newcommand{\FF}{\mathbb{F}}
\newcommand{\EE}{\mathbb{E}}

\newcommand{\cA}{\mathcal{A}}
\newcommand{\cB}{\mathcal{B}}
\newcommand{\cC}{\mathcal{C}}
\newcommand{\cD}{\mathcal{D}}

\newcommand{\cF}{\mathcal{F}}

\newcommand{\cH}{\mathcal{H}}

\newcommand{\cS}{\mathcal{S}}
\newcommand{\cL}{\mathcal{L}}

\begin{document}
\title{Semi-Static and Sparse Variance-Optimal Hedging}
\author{Paolo Di Tella}
\author{Martin Haubold}
\author{Martin Keller-Ressel}
\address{Institute for Mathematical Stochastics, TU Dresden}
\address{Institute for Mathematical Stochastics, TU Dresden}
\address{Institute for Mathematical Stochastics, TU Dresden}
\thanks{MKR thanks Johannes Muhle-Karbe for early discussions on the idea of `variance-optimal semi-static hedging'. We acknowledge funding from the German Research Foundation (DFG) under grant ZUK 64 (all authors) and KE 1736/1-1 (MKR, MH)}
\begin{abstract}
We consider hedging of a contingent claim by a `semi-static' strategy composed of a dynamic position in one asset and static (buy-and-hold) positions in other assets.
We give general representations of the optimal strategy and the hedging error under the criterion of variance-optimality and provide tractable formulas using Fourier-integration in case of the Heston model. We also consider the problem of optimally selecting a sparse semi-static hedging strategy, i.e.\@ a strategy which only uses a small subset of available hedging assets. The developed methods are illustrated in an extended numerical example where we compute a sparse semi-static hedge for a variance swap using European options as static hedging assets.
\end{abstract}

\date{}
\maketitle

\tableofcontents

\section{Introduction}
Semi-static hedging strategies are strategies that are composed of a dynamic (i.e.\@ continuously rebalanced) position in one asset and of static (i.e.\@ buy-and-hold) positions in other assets. Such hedging strategies have appeared in mathematical finance in several different contexts: The hedging of Barrier options (cf.~\cite{carr2011semi}), model-free hedging approaches based on martingale optimal transport (cf.~\cite{beiglbock2013model}), and -- most relevant in our context -- the semi-static replication of variance swaps by Neuberger's formula (cf.~\cite{neuberger1994log}). Compared with fully dynamic strategies, semi-static strategies have the advantage that no rebalancing costs or liquidity risks are associated with the static part of the strategy and hence even assets with limited liquidity can be used as static hedging assets. 

Remarkably, for certain hedging problems, semi-static strategies allow for perfect replication even in incomplete markets -- at least theoretically. Again, the most prominent example is the replication formula for a variance swap, given by \cite{neuberger1994log, CM98}: In any continuous martingale model, a variance swap can be replicated by dynamic hedging in the underlying and a static portfolio of European put- and call-options. This very replication formula is at the heart of the computation of the volatility index \texttt{VIX}, whose value is determined precisely from a discretization of Neuberger's replicating option portfolio (cf.\@ the CBOE's technical document \cite{exchange2014cboe}).

However, Neuberger's result relies on certain idealizations: Most importantly, the static part of the strategy consists of \emph{infinitesimally small} positions in an \emph{infinite number} of puts and calls with strikes ranging from zero to infinity. Any practical implementation of this strategy therefore has to decide on a certain quantization of the theoretical strategy, i.e.\@ how to assign non-infinitesimal weights to the actually tradable put- and call-options. Rather than doing this in an ad-hoc manner, our goal is to determine how to \emph{optimally} implement a semi-static hedging strategy when a finite number $n$ of hedging assets is available. Our optimality criterion is the well-known variance-optimality criterion introduced by \cite{schweizer1984varianten, FS86}, i.e.\@ we minimize the variance of the residual hedging error under the risk-neutral measure. As we show in Section~\ref{sec:varo}, this criterion is pleasantly compatible with semi-static hedging: The semi-static hedging problem separates into an inner problem, which is equivalent to the variance-optimal hedging problem with a single asset (as considered in \cite{schweizer1984varianten, FS86}) and an outer problem, which is an $n$-dimensional quadratic optimization problem, cf.\@ Theorem~\ref{thm:semistatic}.\\

After having analyzed the general structure of the variance-optimal semi-static hedging problem, we turn to another question in Section~\ref{sec:sparse}: How many assets $d < n$ are enough to obtain a `reasonably small' hedging error? In case of Neuberger's formula for the variance swap -- where infinitely many European options reduce the hedging error to zero -- how good is using $12$, $6$ or even just $3$ options? Beyond that, which $3$ options should one select from, say, $30$ that are available in the market? It turns out that this problem of finding a \emph{sparse semi-static hedging} strategy is closely related to the well-known problem of \emph{variable selection} in high-dimensional regression, cf.\@ \cite[Sec.~3.3]{hastie2013elements}, and more generally to sparse modelling approaches in statistics and machine learning.\footnote{These connections to linear regression should not come as a surprise: It has been noted already in \cite{follmer1988hedging} that variance-optimal hedging in discrete time is equivalent to a sequential linear regression problem.}
 Indeed, to solve the problem of optimal selection we will draw from methods developed in statistics, such as the LASSO, greedy forward selection and the method of Leaps-and-Bounds.

Finally, with the goal of a numerical implementation of sparse semi-static hedging in mind, we have to find tractable methods to compute variances and covariances of hedging errors, expressed mathematically as residuals in the Galtchouk-Kunita-Watanabe (GKW) martingale decomposition. Here, we build on results from \cite{KalP07} which allow to calculate the GKW-decomposition `semi-analytically', i.e.\@ in terms of Fourier integrals, in several models of interest, such as the Heston model. The results of \cite{KalP07}, which focus on calculation of the strategy and the hedging error in a classic variance-optimal hedging framwork, are however not sufficient for the \emph{semi-static} hedging problem and we draw from some extensions that are developed in the technical companion paper \cite{ditella2017variance}.

We conclude in Section~\ref{sec:numerics} with a detailed numerical example, implementing the sparse semi-static hedging problem for hedging of a variance swap with put and call options in the Heston model. In particular, we compare the performance of the different solution methods for the subset selection problem, analyze the dependency of optimal hedging portfolio and hedging error on the number $d$ of static hedging assets, and study the influence of the leverage parameter $\rho$ on the optimal solution.

\section{Variance-optimal semi-static hedging}\label{sec:varo}

To set up our model for the financial market, we fix a complete probability space  $(\Om,\cF,\PP)$ equipped with a filtration $\FF$ satisfying the usual conditions. We fix a time horizon $T>0$, assume that $\cF_0$ is the trivial $\sig$-algebra and set $\cF=\cF_T$. We also assume that a state-price density $\frac{\rmd\QQ}{\rmd\PP}$ is given, which uniquely specifies a risk-neutral pricing measure $\QQ$. All expectations $\E{.}$ denote expectations under this risk-neutral measure $\QQ$. We denote by  $S = (S_t)_{t \geq 0}$ the price process of a traded asset, set interest rates to zero to simplify the exposition of results, and assume that $S$ is a continuous square-integrable martingale under $\QQ$. More generally, we denote by $\cH\p{\,2}=\cH\p{\,2}(\FF)$ the set of real-valued $\FF$-adapted square integrable $\QQ$-martingales, which becomes a Hilbert space when equipped with the norm $\|X\|_{\cH\p{\,2}}^2:=\EE\big[X_T^2\big]$. We also set $\cH\p{\,2}_0 := \set{X \in \cH^2: X_0 = 0}$.

\subsection{Variance-optimal hedging}
Before discussing semi-static hedging, we quickly review variance-optimal hedging of a claim $H\p{\,0}$ in $L\p2(\Om,\cF,\QQ)$, up to the time horizon $T > 0$, as discussed e.g. in \cite{FS86}. We identify the claim $H^0$ with the martingale
\[H_t^0 = \Econd{H^0}{\cF_t}, \qquad t \in [0,T]\]
which is an element of $\cH^2$. The set of all admissible dynamic strategies is denoted by 
\[
\Lrm\p2(S):=\set{\vartheta\textnormal{ predictable and $\RR$-valued: } \E{\int_0^T |\vartheta_t|^2 \rmd\aPP{S}{S}_t}<+\infty},
\] 
where $\aPP{S}{S}$ denotes, as usual, the predictable quadratic variation of $S$. The variance-optimal hedge $\vartheta$ with initial capital $c$ of the claim $H^0$ is the solution of 
\begin{equation}\label{eq:var_opt}
\epsilon^2 = \min_{\vartheta \in \Lrm^2(S), c \in \RR}\EE\left[\left(c+\int_0^T \vartheta_t \rmd S_t -H_T^0\right)^2\right].    
\end{equation}
The resulting quantity $\epsilon$ is the minimal hedging error. The minimization problem \eqref{eq:var_opt} can be interpreted as orthogonal projection (in $\cH\p{\,2}$) of the claim $H^0$ onto the closed subspace spanned by deterministic constants (corresponding to the initial capital $c$) and by $\cL\p2(S):=\{\int_0^T \vartheta_t \rmd S_t,\ \vartheta\in\Lrm\p2(S)\} \subset \cH^2_0$, the set of claims attainable with strategies from $\Lrm\p2(S)$. The resulting orthogonal decomposition
\begin{equation}\label{eq:gkw_dec}
H_t^0=c + \int_0^t \vartheta_s \rmd S_s + L_t\,,
\end{equation}
of $H^0$ is known as the Galtchouk-Kunita-Watanabe (GKW-)decomposition of $H^0$ with respect to $S$, cf.\@ \cite{kunita1967square, ansel1993decomposition}. From the financial mathematics perspective, \eqref{eq:gkw_dec} decomposes the claim $H^0$ into initial capital, hedgable risk, and unhedgable residual risk. 

The orthogonality of $L$ to $\cL\p2(S)$ in the Hilbert space sense implies orthogonality of $L$ to $S$ in the martingale sense, i.e.\@ it holds that $\aPP{L}{S}=0$. Hence, the variance-optimal strategy $\vartheta$ can be computed from \eqref{eq:gkw_dec} as
\begin{equation}\label{eq:gkw_dec.int}
\aPP{H}{S}_t =\int_0^t \vartheta_s\,\rmd \aPP{S}{S}_s\,,
\end{equation}
and $\vartheta$ can be expressed as the Radon-Nikodym derivative $\vartheta=\rmd\aPP{H}{S}\slash\rmd\aPP{S}{S}$.

\subsection{The variance-optimal semi-static hedging problem}

We are now prepared to discuss the variance-optimal \emph{semi-static} hedging problem and its solution. In addition to the contingent claim $H^0$ which is to be hedged, denote by $H=(H^1, \dotsc, H^n)\p\top$ the vector of supplementary contingent claims, all assumed to be square-integrable random variables in $L\p2(\Om,\cF,\QQ)$. Again, we associate to each $H^i$ the martingale
\begin{equation}\label{eq:ass.mart}
H^i_t := \Econd{H^i}{\cF_t},\qquad t\in[0,T],\quad i=0,\ldots,n.
\end{equation} 
The static part of the strategy can be represented by an element $v$ of $\RR^n$, where $v_i$ represents the quantity of claim $H^i$ bought at time $t = 0$ and held until time $t = T$. The dynamic part $\vartheta$ of the strategy is again represented by an element of $\Lrm\p2(S)$. 
\begin{definition}[Variance-Optimal Semi-Static Hedging Problem]
The variance-optimal semi-static hedge $(\vartheta, v) \in \Lrm^2(S) \times \RR^n$ and the optimal initial capital $c \in \RR$ are the solution of the minimization problem 
\begin{equation}\label{eq:setup}
\epsilon^2 = \min_{(\vartheta, v) \in \Lrm^2(S) \times \RR^n, c \in \RR}\EE\left[\Big(c-v^\top\E{H_T} + \int_0^T \vartheta_t \rmd S_t -(H^0_T-v^\top H_T)\Big)^2\right].    
\end{equation}
\end{definition}
Note that $v^\top\E{H_T}$ is the cost of setting up the static part of the hedge and its terminal value is $v^\top H_T$. The dynamic part is self-financing and results in the terminal value $\int_0^T \vartheta_t \rmd S_t$. Adding the initial capital $c$ and subtracting the target claim $H^0_T$ yields the above expression for the hedging problem.

To solve the variance-optimal semi-static hedging problem, we decompose it into an inner and an outer minimization problem and rewrite \eqref{eq:setup} as
\begin{equation}\label{eq:inner_outer}
\begin{cases}
\epsilon^2(v) = \min_{\vartheta \in \Lrm^2(S), c \in \RR}\EE\left[\Big(c-v^\top\E{H_T}+\int_0^T \vartheta_t \rmd S_t - (H^0_T-v^\top H_T)\Big)^2\right], & \quad \text{(inner prob.)}\\    
\epsilon^2 = \min_{v \in \RR^n} \epsilon(v)^2. &\quad \text{(outer prob.)}\\
\end{cases}
\end{equation}
The inner problem is of the same form as \eqref{eq:var_opt}, while the outer problem turns out to be a finite dimensional quadratic optimization problem. To formulate the solution, we write the GWK-decompositions of the claims ($H^0, \dotsc, H^n)$ with respect to $S$ as
\begin{equation}\label{eq:gkw_dec.H}
H\p{\,i}_t= H\p{\,i}_{\,0}+ \int_0^t \vartheta_s^i \rmd S_s + L\p{\,i}_t, \quad i =0, \dots, n.
\end{equation}
Similarly to \eqref{eq:gkw_dec.int} we get
\begin{equation}\label{eq:int.gkw.dec.H}
\vartheta\p{\,i} = \frac{\rmd\aPP{H\p{\,i}}{S}}{\rmd{\aPP{S}{S}}}\,,\quad i=0,\ldots,n.
\end{equation}
and introduce the vector notation $\vartheta:=(\vartheta\p{\,1},\ldots,\vartheta\p{\,n})\p\top$ for the strategies and $L:=(L\p{\,1},\ldots,L\p{\,n})\p\top$ for the residuals in the GKW-decomposition. Finally we formulate the following condition:
\begin{definition}[Non-redundancy condition]
The supplementary claims $H = (H^1, \dots, H^n)\p\top$ satisfy the \emph{non-redundancy condition} if there is no $x \in \RR^n \setminus \set{0}$ with $x^\top L_T = 0$, a.s. 
\end{definition}
Intuitively, existence of a non-zero $x$ with $x^\top L_T = 0$, means that the number of supplementary assets can be reduced without changing the hedging error $\epsilon^2$ in \eqref{eq:setup}. We are now prepared to state our main result on the solution of the variance-optimal semi-static hedging problem:

\begin{theorem}\label{thm:semistatic} Consider the variance-optimal semi-static hedging problem \eqref{eq:setup} and set
\begin{align}\label{not:var.cov}
A:=\Var[L_{\,T}\p{\,0}],\qquad B:=\Cov[L_{\,T}, L_{\,T}\p{\,0}],\qquad C:=\Cov[L_{\,T}, L_{\,T}]\,.
\end{align}
Under the non-redundancy condition, $C$ is invertible and the unique solution of the semi-static hedging problem is given by
\[c =\E{H^0_T}, \qquad v = C^{-1} B,\qquad \vartheta^v = \vartheta^0 - v^\top \vartheta.\]
The minimal squared hedging error is given by
\[\varepsilon^2 = A - B^\top C^{-1} B.\]
Moreover, the elements of $A$, $B$ and $C$ can be expressed as
\begin{equation}
\label{eq:ex.var.cov}
\EE\big[L_{\,T}^i L_{\,T}^j\big] = \E{\scal{L^i}{L^j}_T} = \E{\scal{H^i}{H^j}_T - \int_0^T  \vartheta_t^i \vartheta_t^j \, \rmd\scal{S}{S}_t},\quad i,j=0,\dotsc,n\,.
\end{equation}
\end{theorem}

\begin{corollary}\label{cor:semistatic}
If the non-redundancy condition does not hold true, then any solution $v \in \RR^n$ of the linear system $Cv = B$, together with $c =\E{H^0_T}$ and $\vartheta^v = \vartheta^0 - v^\top \vartheta$ is a solution of the semi-static hedging problem. The solution set is never empty, and the solution which minimizes the Euclidian norm of $v$ can be obtained by setting $v = C^\dagger B$, where $C^\dagger$ denotes the Moore-Penrose pseudo-inverse of $C$. 
\end{corollary}

Notice that the minimal squared hedging error $\varepsilon^2 $ in Theorem~\ref{thm:semistatic} is the Schur complement of the block $C$ in the `extended covariance matrix'
\[
\Cov[(L_T^0, L_{\,T}), (L_T^0, L_{\,T})] = \begin{bmatrix}  A& B^\top \\ B &C  \end{bmatrix}.
\]
In particular, if $(L_T^0, L_T)$ has normal distribution, then $\varepsilon^2$ can be expressed as $\varepsilon^2 = \Var\left[\left.L_{\,T}\p{\,0}\right|L_T\right]$.

\begin{proof}[Proof of Theorem~\ref{thm:semistatic} and Corollary~\ref{cor:semistatic}]
First, we consider the inner minimization problem in \eqref{eq:inner_outer}. This problem is equivalent to the variance-optimal hedging problem for the claim $H^v := H^0 - v^\top H$. The solution $\vartheta^\nu$ is given by the GKW-decomposition
\begin{equation}\label{eq:gkw_dec.lin}
(H^0_t-v^\top H_t) = (H^0_0 - v^\top H_0) + \int_0^t \vartheta_s^v \rmd S_s + L^v_t, \quad t \in [0,T]
\end{equation}
of the martingale $(H\p{\,0}-v\p\top H)$ with respect to $S$. By \eqref{eq:gkw_dec.int} we obtain
\[\vartheta^v_t  = \frac{\rmd\aPP{(H^0 -v^ \top H)}{S}_t}{\rmd\aPP{S}{S}_t} = \vartheta^0_t - v^\top \vartheta_t,\]
using the bilinearity of the predictable quadratic covariation. Uniqueness of the GKW-decomposition yields
$L^v_t= L^0_t - v^\top L_t$ and the squared hedging error is given by
\begin{align*}
\varepsilon(v)\p{\,2} &= \EE\big[(L\p{\,v})\p2\big] = \EE\big[(L\p{\,0}_{\,T} - v\p\top{L}_{\,T})\p2\big]  
= v^\top\EE\big[L_{\,T} L_{\,T}\p\top\big] v - 2 v\p\top  \EE\big[L_{\,T} L_{\,T}\p{\,0}\big] + \EE\big[(L_{\,T}\p{\,0})\p2]  = \\
&= v^\top C v - 2v^\top B + A.
\end{align*}
Thus, the outer optimization problem in \eqref{eq:inner_outer} becomes
\begin{equation}\label{eq:outer_ABC}
\varepsilon^2 =  \min_{v \in \RR^n} \Big(v^\top C v - 2v^\top B + A\Big).
\end{equation}
Since $C$ is positive semi-definite, the first order-condition $Cv = B$ is necessary and sufficient for optimality of $v$. Under the non-redundancy condition, $\Var(x^\top L_T) > 0$ for any $x \in \RR^n \setminus \set{0}$, hence $C$ is positive definite and in particular invertible. The unique solution of the outer problem is therefore given by $v = C^{-1}B$, completing the proof of Theorem~\ref{thm:semistatic}.

For the corollary, it remains to show that $Cv = B$ has a solution, even when the non-redundancy condition does not hold. A solution exists, if $B$ is in the range of $C$, or equivalently, if $B$ is in $(\ker\,C)^\bot$. By assumption $\ker\,C$ is non-empty, and we can choose come $x \in \ker\,C$, i.e.\@ with $x^\top C = 0$. Since $C$ is the covariance matrix of $L_T$ is follows that $x^\top L_T = 0$, a.s. This implies that also $x^\top B = \Cov(x^\top L_T,L_T^0) = 0$, for all $x \in \ker\,C$ and hence that $B \in (\ker\,C)^\bot$. 
\end{proof}

Finally, we compute the \emph{hedge contribution} of a single supplementary asset $H^{n+1}$. By hedge contribution, we mean the reduction in squared hedging error that is achieved by adding the asset $H^{n+1}$ to a given pool of supplementary assets $(H^1, \dotsc, H^n)$. We denote by $\varepsilon_n\p2$ and $\varepsilon_{n+1}\p 2$ the minimal hedging error achieved with supplementary assets $(H^1, \dotsc, H^n)$ and $(H^1, \dotsc, H^{n+1})$ respectively. 

\begin{proposition}[Relative Hedge Contribution]\label{prop:hedge_contribution}
Suppose that the non-redundancy condition holds true for all supplementary assets $H^1, \dotsc, H^{n+1}$. Then the relative hedge contribution $RHC_{n+1}$ of $H^{n+1}$ is given by
\begin{equation}\label{eq:est.err.pl.one}
RHC^2_{n+1} := \frac{\varepsilon_{n}\p 2- \varepsilon_{n+1}\p 2}{\varepsilon_{n}\p 2} =  \frac{\Big(\Cov[L^{n+1}_T,L^0_T] - K^\top C^{-1} B\Big)^2}{\left(\Var[L^{n+1}_T]- K^\top C^{-1} K\right)\left(A - B^\top C^{-1} B\right)} \in [0,1],
\end{equation}
where $K \in \RR^n$ with $K_i = \Cov(L_T^i, L_T^{n+1})$ for $i=1,\dotsc,n$. 
\end{proposition}
\begin{remark}
The expression for the relative hedge contribution has an intuitive interpretation under the assumption that the residuals $(L_T^0, \dotsc, L_T^{n+1})$ have multivariate normal distribution. In this case the hedge contribution of $H^{n+1}$ is equal to the \emph{partial correlation} $\Cor(L_T^0, L_T^{n+1} | L_T)$ of $L_T^0$ and $L_T^{n+1}$, given $L_T$, cf.~\cite[Ch.~5.3]{muirhead2009aspects}. Thus, roughly speaking, a supplementary asset has a high hedge contribution, if it is strongly correlated with $H^0$, even after conditioning on all claims that are attainable with semi-static strategies in $S$ and $(H^1, \dots, H^n)$.
\end{remark}
\begin{proof}
We set $B^{\,\textnormal{new}}:=[B\p\top, \Cov[L_T\p{n+1},L_T\p0]]\p\top$;  $C^{\,\textnormal{new}}:=(C^{\,\textnormal{new}}_{i,j})_{i,j=1,\ldots,n+1}$, where $C^{\,\textnormal{new}}_{i,j}:=\Cov[L\p i_T,L\p j_T]$, $i,j=1,\ldots,n+1$. Then we have
\[
C^{\,\textnormal{new}} = \begin{bmatrix} C & K \\ K^\top & \Var[L^{n+1}_T] \end{bmatrix} \,
\]
which is invertible due to the non-redundancy condition. Write $M = (\Var[L^{n+1}_T]-K\p\top C\p{\,-1} K)$ for the Schur complement of $C$ in $C^{\,\textnormal{new}}$. Using the Schur complement, the inverse of the block matrix $C^{\,\textnormal{new}}$ can be written as
\[
(C^{\,\textnormal{new}})^{\,-1} =   \begin{bmatrix} C^{\,-1} +  C^{\,-1} K K^\top C^{\,-1}M^{\,-1} & -C^{\,-1} K M^{\,-1} \\ -K^\top C^{\,-1}M^{\,-1} & M^{\,-1} \end{bmatrix} ,
\]
cf.~\cite{horn2012matrix}, and applying Theorem~\ref{thm:semistatic} yields $\varepsilon_{n+1}(u)\p2=A-(B^{\,\textnormal{new}})\p\top(C^{\,\textnormal{new}})^{\,-1}B^{\,\textnormal{new}}$. Consequently
\[
\varepsilon_d(u)^2 - \varepsilon_{n+1}(u)^2 = (B^\textnormal{new})^\top (C^\textnormal{new})^{-1} B^\textnormal{new} - B^\top C^{-1} B\]
and further algebraic manipulation yields \eqref{eq:est.err.pl.one}.
\end{proof}

\subsection{The variance swap and long/short constraints} \label{sub:swap}
We review the semi-static hedging of a variance swap with an infinite pool of European put- and call-options, as discussed in \cite{neuberger1994log, CM98}. We will apply the methods developed in this paper to this hedging problem in section~\ref{sec:numerics}. It also serves as a motivation to add long/short constraints to the semi-static hedging problem.

Recall that a \emph{variance swap} is a contingent claim on an underlying traded asset $S$, which at maturity $T$ pays an amount $H^\text{swap}_T:=[\log S,\log S]_T-k$, where $k\in\RR$. Usually, $k$ is chosen such that the value of the contract is zero at inception, and the corresponding value $k_* = \E{[\log S, \log S]_T}$ is called the \emph{swap rate}. Recall that our only assumption on the discounted price process $S$ is that it is a square-integrable strictly positive continuous martingale. Applying It\^o's formula to $\log S_T$ we get
\begin{equation}\label{eq:it.log}
\log S_T=\log S_0 + \int_0 ^T \frac{1}{S_t} \,\rmd S_t -\frac{1}{2} \int_0^T \frac{1}{S_t^2} \rmd[S,S]_t
\end{equation}
 Hence
\[H^\text{swap}_T = \int_0^T \frac{1}{S_t^2} \rmd[S,S]_t - k = 2  \int_0^T \frac{1}{S_t} \,\rmd S_t - 2 \log\frac{S_T}{S_0} - k
\]  
that is, to hedge the variance swap, it is enough to dynamically trade in the stock $S$ and enter a static position in the `log-contract' with payoff $\log\frac{S_{\,T}}{S_{\,0}}$, cf.\@ \cite{neuberger1994log}. Furthermore,  from \cite{CM98}, we have
\begin{equation}\label{eq:rep.log.con}
\log \frac{S_{\,T}}{S_{\,0}}=\frac{S_{\,T}-S_{\,0}}{S_{\,0}}-\int_0\p {S_{\,0}}\frac{(K - S_{\,T})\p+}{K\p{\,2}}\,\rmd  K- \int_{S_{\,0}}\p{\infty}\frac{(S_{\,T} - K)\p+}{K\p{\,2}}\,\rmd  K\,,
\end{equation}
which inserting into the above equation yields
\begin{equation}\label{eq:swap_replication}
H^\text{swap}_T=2 \int_0^T \left(\frac{1}{S_t} - \frac{1}{S_0}\right) \rmd S_t  - k + 2 \int_0\p {S_{\,0}}\frac{(K - S_{\,T})\p+}{K\p{\,2}}\,\rmd  K + 2 \int_{S_{\,0}}\p{\infty}\frac{(S_{\,T} - K)\p+}{K\p{\,2}}\,\rmd  K.
\end{equation}
  This equality can be interpreted as a semi-static replication strategy for the variance swap, which uses a dynamic position in $S$ and a static portfolio of \emph{infinitesimally small} positions in an \emph{infinite number} of out-of-the-money puts and out-of-the-money calls. We make several observations:
\begin{enumerate}[(a)]
\item For any practical implementation the `infinitesimal portfolio' has to be discretized and portfolio weights have to be assigned to each put and call.
\item Since they are calls and puts on the same underlying asset, the static hedging assets are highly correlated. 
\item The static positions in puts and calls are long positions only.
\end{enumerate}
To address point (a) different ad-hoc discretizations of the integrals in \eqref{eq:swap_replication} are possible (e.g. left or right Riemann sums, trapezoidal sums, etc.\@). However, it is not obvious which discretization is optimal in the sense of minimizing the hedging error. The choice of an optimal discretization in the variance-minimizing sense is precisely given by Theorem~\ref{thm:semistatic}.

Point (b) suggests that given a moderate number (say 30) of puts and calls as static hedging assets, many of them will be redundant in the sense that their hedge contribution (given the other supplementary assets) is small. This observation motivates the sparse approach of the next section and will be confirmed numerically in the application Section~\ref{sec:numerics}.

Point (c) finally motivates the addition of short/long constraints, or more generally, linear constraints of the type
\begin{equation}\label{eq:constraint}
v^\top p  \ge 0, 
\end{equation}
where $p \in \R^n$ is fixed, to the outer problem in \eqref{eq:inner_outer}. With these constraints, the outer problem is a linearly constrained quadratic optimization problem, which can still be efficiently solved by standard numerical software.

\section{Sparse semi-static hedging}\label{sec:sparse}
We now focus on the problem of optimal selection of static hedging assets, as outlined in the introduction and motivated in the previous section. Note that the subset selection only affects the static part of the strategy and hence only the outer problem in \eqref{eq:inner_outer}. Recall the $\ell_1$-norm $\norm{v}_1 = \sum_{i=1}^n |v_i|$ on $\RR^n$ and the (non-convex) $\ell_0$-quasinorm $\norm{v}_0$ which counts the number of non-zero elements of $v$, cf.\@ \cite{foucart2013mathematical}.
\begin{definition}[Sparse Variance-Optimal Semi-Static Hedging Problem]
The \emph{sparse variance-optimal semi-static hedge} $(\vartheta, v) \in \Lrm^2(S) \times \RR^n$ with effective portfolio size $d < n$ and its optimal initial capital $c \in \RR$ are the solution of the minimization problem 
\eqref{eq:inner_outer}, with the outer problem replaced by 
\begin{align}\label{eq:l0}
\epsilon^2 &= \min_{v \in \RR^n, v\geq 0} (v^\top C v  - 2v^\top B + A), \quad \text{subj. to}\quad \norm{v}_0 \le d.\qquad &&\text{($\ell_0$-constrained problem)}\\
\intertext{The $\ell_1$-relaxation of this problem is given by}
\epsilon^2 &= \min_{v \in \RR^n, v\geq 0} (v^\top C v  - 2v^\top B + A) + \lambda \norm{v}_1, \qquad &&\text{($\ell_1$-relaxation)}.\label{eq:l1}
\end{align}
where $\lambda > 0$ is a tuning parameter that replaces $d$. In both problems, we allow for long/short contains of the form \eqref{eq:constraint}.
\end{definition}

Of course, the minimization problem \eqref{eq:l0} is equivalent to the extensively studied  subset selection problem in linear regression and \eqref{eq:l1} to its convex relaxation in Lagrangian form, usually called LASSO. We refer to \cite{hastie2013elements} for a general overview and to \cite{tibshirani1996regression} for the LASSO. We emphasize that 
\begin{itemize}
\item The $\ell_0$-constrained subset selection problem \eqref{eq:l0} is non-convex and hard to solve exactly if the dimension $n$ is high.
\item The $\ell_1$-penalized minimization problem \eqref{eq:l1} is convex and efficient numerical solvers exist even for large $n$. Its solution is usually a good approximation to the exact subset selection problem, but no guarantee of being close to the solution of \eqref{eq:l0} can be given in general.\footnote{We remark that conditions for the perfect recovery of solutions of \eqref{eq:l0} by solving \eqref{eq:l1} can be given within the theoretical framework of compressive sensing, see e.g. \cite{foucart2013mathematical}.}
\end{itemize}

To illustrate the effect of the $\ell_1$-penalty, denote by $v_*$ the solution of the unpenalized hedging problem \eqref{eq:setup} and assume for a moment that all GKW-residuals $(L^0_T, \dotsc, L^d_T)$ are uncorrelated. This assumption is highly unrealistic in the hedging context, but leads to a simple form of the solution of the penalized problem, cf.\@ \cite{tibshirani1996regression}: It is given by $v' = \mathrm{sign}(v) (|v| - \lambda)^+$, i.e.\@ all static positions are shrunk towards zero by $\lambda$ and truncated when zero is reached. This nicely illustrates the sparsifying effect of the penalty and the role of $\lambda$.

While \eqref{eq:l1} is frequently used as a surrogate for \eqref{eq:l0}, the following alternatives exist for solving \eqref{eq:l0} directly, or for approximating its solution. Again, we refer to \cite{tibshirani1996regression} for further details on the described methods:
\begin{description}
\item[Brute-Force] Solve the quadratic optimization problem for each possible subset of cardinality $d$. Since there are $d \choose n$ of these subsets, this approach is usually not efficient and becomes completely infeasible for large $n$.
\item[Leaps-and-Bounds] `Leaps-and-Bounds' is a branch-and-bound algorithm introduced by \cite{furnival1974regressions} for subset selection in linear regression, which gives an exact solution to \eqref{eq:l0} without testing all possible subsets. 
\item[Greedy Forward Selection] A simple greedy approximation to \eqref{eq:l0} is to assume that the optimal subsets of different cardinality are nested. In the forward approach the problem \eqref{eq:l0} is first solved for $d=1$, which is easy. Then, iteratively, the supplementary claim with the largest relative hedge contribution (see \eqref{prop:hedge_contribution}) is added to the set of active static positions in each step. The same procedure could be used backwards (`greedy backward selection')  i.e.\@ starting with $d=n$ and then removing iteratively the supplementary claim with the smallest hedge contribution. In general, no guarantee of being close the exact solution of \eqref{eq:l0} can be given for these methods.
\end{description}

We will compare the practical performance of the different solution methods in Section~\ref{sec:numerics}.

\section{Stochastic volatility models with Fourier representation}
The final ingredient that is still missing for a numerical solution of the (sparse) semi-static hedging problem is an efficient method to compute the quantities $A$, $B$ and $C$ from Theorem~\ref{thm:semistatic}. One possible approach  would be to compute transition densities of $S$ and $H^0, \dotsc, H^n$ by Monte-Carlo-Simulation and to compute the GKW-decomposition by sequential backward regression, cf.\@ \cite{follmer1988hedging}. Due to the fact that the joint distribution of $S$ and all price processes of supplementary claims is needed, we expect a heavy computational load in order to obtain reasonably large accuracy with this method. An interesting alternative method has been suggested by \cite{KalP07} (see also \cite{hubalek2006variance, pauwels2007variance}) for the classic variance-optimal hedging problem \eqref{eq:var_opt} of European claims. This alternative is based on the well-known Fourier method for pricing of European claims, cf.\@ \cite{CM98, raible2000levy, KalP07}. 
\subsection{Fourier representation of strategies and hedging errors}

We stay close to the framework of \cite{KalP07} and assume that the payoff of some option $H$ is given by $H = f(X_T)$, where $X$ is the log-price process of the underlying stock, i.e.{} we also assume $S = \exp(X)$. The payoff of a call for example can be written as $f(x) = (e^x - K)^+$, but it is not necessary to restrict to this specific case. Furthermore, we assume that the (rescaled) two-sided Laplace transform 
\begin{equation}\label{eq:rep.fourier}
\tilde f(u)= \frac1{2\pi\,\rmi}\,\int_{-\infty}\p{+\infty}\exp(-\,u x) f(x)\rmd x
\end{equation}  
of the payoff exists at some $u = R \in \RR$ and is integrable over the strip 
\[\cS(R) :=\{u\in\CC: \Re u = R\}\] in the complex plane. If the integrability condition $\E{e^{R X_T}} < \infty$ holds, then the risk-neutral price of the claim $H$ at time $t \in [0,T]$ can be recovered by the Fourier-type integral
\begin{equation}\label{eq:fourier}
H_t=\int_{\cS(R)} H_t(u) \tilde f(u) \rmd  u,
\end{equation}
along $\cS(R)$, where we denote the conditional moment generating function (analytically extended to the complex plane) of $X_T$ by 
\[H_t(u) := \Econd{e^{u X_T}}{\cF_t}.\]
Note that $H_t(u)$ is well-defined on $\cS(R)$ due to the integrability condition imposed on $X_T$. In the important cases of European puts and calls, the two-sided Laplace transform $\tilde f$ is given by
\[\tilde f(u) = \frac{1}{2 \pi\,\rmi} \frac{K^{1-u}}{u(u-1)},\]
with $R > 1$ for calls and $R < 0$ for puts, cf.\@ \cite[Sec.~4]{hubalek2006variance}.

The key insight, pioneered by \cite{hubalek2006variance} for variance-optimal hedging in models with independent increments and by \cite{KalP07, pauwels2007variance} for affine stochastic volatility models, is that the Fourier representation \eqref{eq:fourier} of European claims can be extended to their GKW-decomposition \eqref{eq:gkw_dec}. More precisely, both the strategy $\vartheta$ and the hedging error $\epsilon^2 = \E{L_T^2}$ of the variance-optimal hedging problem \eqref{eq:var_opt} can be expressed in terms of Fourier-type integrals, similar to \eqref{eq:fourier}. For our problem of interest, the \emph{semi-static} hedging problem \eqref{eq:setup}, the results of \cite{hubalek2006variance, KalP07, pauwels2007variance} are not sufficient: To obtain the quantities $A$, $B$ and $C$ of Theorem~\ref{thm:semistatic}, we also need to compute the covariances $\EE[L_T^i L_T^j]$ between the GKW-residuals of different claims. In the companion paper~\cite{ditella2017variance} we extend the results of \cite{hubalek2006variance, KalP07, pauwels2007variance} 
to the semi-static hedging problem. Moreover, we show that the method can be used in any stochastic volatility models where the Fourier-transform of the log-price $X$ is known (e.g. the Heston, the 3/2 or the Stein-Stein model, cf.\@~\cite{lewis2000option}). Here, we only need a special case of the more general results in \cite{ditella2017variance}, which is condensed into Theorem~\ref{thm:markov}(i) below. 

In order to formulate the representation result, we assume that a claim $H^0$ (e.g. a variance swap), and supplementary assets $H^1, \dotsc, H^n$ with Fourier representations \eqref{eq:rep.fourier} are given, and define for $u, u_1, u_2 \in \CC$ complex-valued predictable processes of finite variation $\cA, \cB(u), \cC(u_1,u_2)$ by
\begin{subequations}\label{eq:complex_proc_ABC}
\begin{align}
\rmd\cA&= \rmd \scal{H^0}{H^0} - \vartheta^0\vartheta^0 \rmd \scal{S}{S}, &&\quad \cA_0 = 0\\
\rmd\cB(u) &= \rmd \scal{H^0}{H(u)} - \vartheta^0\vartheta(u) \rmd \scal{S}{S}, &&\quad \cB_0(u) = 0,\\
\rmd\cC(u_1,u_2) &= \rmd \scal{H(u_1)}{H(u_2)} - \vartheta(u_1)\vartheta(u_2) \rmd \scal{S}{S}, &&\quad \cC_0(u_1,u_2) = 0. 
\end{align}
\end{subequations}

\begin{theorem}\label{thm:markov}Let a stochastic volatility model with forward price process $S = e^X$ and variance process $V$ be given, and let $T>0$ be a fixed time horizon.  Let $H^0$ be a variance swap with payoff $[X,X]_T = \int_0^T V_t \rmd t$ and let the supplementary assets $(H^1, \dotsc, H^n)$ be European puts or calls with Fourier representations given by \eqref{eq:rep.fourier}. Assume that $(S,V)$ are continuous square-integrable semi-martingales and that there exist functions $h(u,t,V_t)$, $\gamma(t,V_t)$, continuously differentiable in the last component, such that
\begin{equation}\label{eq:markov}
H_t(u) = \Econd{e^{uX_T}}{\cF_t} = e^{uX_t} h(u,T-t,V_t), \qquad F_t := \Econd{[X,X]_T - [X,X]_t}{\cF_t} = \gamma(T-t,V_t).
\end{equation}
Then the following holds true:
\begin{enumerate}[(i)]
\item The quantities $A$, $B$ and $C$ in Theorem~\ref{thm:semistatic} can be represented as
\begin{subequations}\label{eq:complex_ABC}
\begin{align}
A &= \E{\cA_T} &&\\
B_i &= \int_{\cS(R_i)} \E{\cB_T(u)} \tilde f_i(u) \rmd u,  &&\quad i \in \set{1, \dotsc, n} \\
C_{ij} &= \int_{\cS(R_i)} \int_{\cS(R_j)} \E{\cC_T(u_1,u_2)} \tilde f_i(u_1) \tilde f_j(u_2) \rmd u_1 \rmd u_2, &&\quad (i,j) \in \set{1, \dotsc, n}^2. \label{eq:complex_C}
\end{align}
\end{subequations}
\item The processes \eqref{eq:complex_proc_ABC} can be written as
\begin{subequations}\label{eq:complex_ABC_markov}
\begin{align}
\rmd\cA_t&= (\partial_v \gamma(T-t,V_t))^2 \,dQ_t, \\
\rmd\cB_t(u) &= e^{uX_t} \,\partial_v h(u,T-t,V_t)\, \partial_v \gamma(T-t,V_t)\,dQ_t\\
\rmd\cC_t(u_1,u_2) &= e^{(u_1 + u_2)X_t} \,\partial_v h(u_1,T-t,V_t)\, \partial_v h(u_2,T-t,V_t) \,dQ_t
\end{align}
\end{subequations}
where
\begin{equation}
dQ = \rmd [V,V] - \frac{\rmd[X,V]}{\rmd [X,X]} \rmd[X,V].
\end{equation}
\end{enumerate}
\end{theorem}

\begin{proof} Part (i) of the theorem is technically demanding and follows from Theorems~4.5, 4.6 and 4.8 in the companion paper~\cite{ditella2017variance}. In order to show part (ii), let $Y$ be a $\RR^n$-valued continuous semi-martingale and let $\alpha, \beta$ be functions in $C^2(\RR^n,\CC)$. Using Ito's formula (cf.\@ \cite[Thm.~I.4.57]{JS00}) and the properties of quadratic covariation (cf.\@ \cite[Thm~I.4.49]{JS00}) we obtain the calculation rule
\begin{equation}\label{eq:rule}
\rmd [\alpha(Y),\beta(Y)] = \sum_{i,j} \partial_{y_i} \alpha(Y) \partial_{y_j} \beta(Y) \rmd [Y_i, Y_j].
\end{equation}
Inserting the definition of the variance-optimal strategy $\vartheta(u) = \frac{\rmd \scal{H(u)}{S}}{\scal{S}{S}}$ into \eqref{eq:complex_proc_ABC} and recognizing that for continuous martingales predictable variation $\scal{.}{.}$ and quadratic variation $[.,.]$ coincide, we obtain
\[dC(u_1,u_2) = \rmd [H(u_1),H(u_2)] - \frac{\rmd[H(u_1),S]}{\rmd[S,S]} \rmd [H(u_2),S]\]
for $C$ and similar expressions for $B$ and $A$. Using assumption \eqref{eq:markov} and applying \eqref{eq:rule} several times we obtain \eqref{eq:complex_ABC_markov}.
\end{proof}

\subsection{The Heston model} In the Heston model (cf.\@ \cite{heston1993closed}) the risk neutral price process $S$ is given by $S_{\,t}=S_{\,0}\exp(X_{\,t})$, $t\geq 0$, where
\begin{subequations}\label{eq:Heston}
\begin{align}
\rmd X_{\,t} &= -\frac12 V_{\,t}\rmd t+\sqrt{V_{\,t}} \rmd W_{\,t}\p{\,1},\\
\rmd V_{\,t} &=-\lambda(V_{\,t}-\kappa )\rmd t+\sigma\sqrt{V_{\,t}} \rmd W_{\,t}\p{\,2}\,,
\end{align}
\end{subequations}
where  $W\p{\,1}$ and $W\p{\,2}$ are two Brownian motions such that $\aPP{W\p{\,1}}{W\p{\,2}}_{\,t}=\rho t$, $\rho\in[-1,1]$; $\lambda,\sigma, \kappa>0$. \\

The joint moment generating function of the Heston model is known explicitly and of the form
\begin{equation}\label{eq:joi.char}
\E{\exp\left(uX_T + wV_T\right)}=\exp\left(\phi_T(u,w)+\psi_T(u,w) V_0 +u X_0\right),
\end{equation}
well-defined for real arguments in the set
\begin{equation}\label{eq:real_moments}
\cD_T := \set{(u,w) \in \RR^2: \E{\exp\left(uX_T + wV_T\right)} < \infty},
\end{equation}
and with analytic extension to the associated `complex strip' 
\[S(\cD_T) := \set{(u,w) \in \CC^2: (\Re u , \Re w) \in \cD_T}.\]
To represent $\phi_t(u,w)$ and $\psi_t(u,w)$, we introduce $\Delta(u) = (\rho \sigma u - \lambda)^2 - \sigma^2(u^2 - u)$, 
\[r_\pm = r_\pm(u,w) := \frac{1}{\sigma^2} \left(\lambda - \rho \sigma u \pm \sqrt{\Delta(u)}\right) \qquad g = g(u,w) = \frac{r_- - w}{r_+ - w}.\]
Then the explicit expression of $\psi_t$ is given, for $(u,w) \in \cD_t$ by (cf.\@ \cite[Prop.~4.2.1]{Al15}), 
\begin{equation}\label{eq:sol.psi}
\psi_t(u,w):=\begin{cases}
w+(r_- - w)\,\frac{1-\exp\big(- t\sqrt{\Delta}\big)}{1-g\exp(- t\sqrt{\Delta})},&\quad \Delta(u) \neq0;\\\\
w+(r_- - w)^2\,\frac{\sigma^2 t}{2+\sigma^2t (r_- - w)},&\quad \Delta(u) = 0\,.
\end{cases}
\end{equation}
with the convention
\[\frac{\exp(-t\sqrt\Delta)-g}{1-g}:=1,\qquad \frac{1-\exp(t\sqrt\Delta)}{1-g\exp(t\sqrt\Delta)}:=0\]
 whenever the denominator of $g$ is equal to zero. 
Moreover, $\phi_t(u,w)$ is given by
\begin{equation}\label{eq:sol.phi}
\phi_t(u,w):=\begin{cases}
\lambda \kappa  r_- t  - \frac{2\lambda \kappa}{\sigma^2} \log \left(\frac{1 - g \exp(-t \sqrt{\Delta})}{1 - g}\right)&\quad \Delta(u) \neq0;\\\\
\lambda \kappa r_- t  - \frac{2\lambda \kappa}{\sigma^2} \log \left(1 + \frac{\sigma^2}{2}(r_- - w)t \right)&\quad \Delta(u) = 0\,.
\end{cases}
\end{equation}
The following theorem specializes Theorem~\ref{thm:markov} to the Heston model and gives (up to integration) explicit expressions for the quantities $A$, $B$ and $C$ from Theorem~\ref{thm:semistatic}. The proof of the theorem is given in Appendix~\ref{app:heston}.
\begin{theorem}\label{thm:heston}Let $(X,V)$ be given by the Heston model \eqref{eq:Heston} and let the claim $H^0$ be a variance swap, i.e.\@ with payoff $H^0_T = [X,X]_T$ at maturity $T$. Let the supplementary claims $(H^1, \dotsc, H^n)$ be European puts and calls with payoffs $f_i$ and two-sided Laplace transforms $\tilde f_i$, integrable along strips $\cS(R_i)$, as in \eqref{eq:fourier}. If $\E{e^{2R_i X_T}} < \infty$ for all $i = 1, \dotsc, n$ then the quantities $A$, $B$ and $C$, defined in Theorem~\ref{thm:semistatic} are given by
\begin{align*}
A &= \frac{\sigma^2 (1 - \rho^2)}{\lambda^2}  \int_0^T  \left(1 - e^{-\lambda(T-t)}\right)^2 \E{V_t} \,dt\\
B_i &= \frac{\sigma^2 (1 - \rho^2)}{\lambda}  \int_0^T \int_{\cS(R_i)} \left(1 - e^{-\lambda(T-t)}\right) \psi_{T-t}(u) \E{H_t(u) V_t} \tilde f_i(u) \,du \,dt\\
C_{ij} &= \sigma^2 (1 - \rho^2) \cdot \\
&\phantom{=}\int_0^T \int_{\cS(R_i)} \int_{\cS(R_j)}  \psi_{T-t}(u_1) \psi_{T-t}(u_2) \E{H_t(u_1,u_2) V_t} \tilde f_i(u_1) \tilde f_j(u_2) \,du_1 \,du_2 \,dt,
\end{align*}
where 
\begin{subequations}\label{eq:expectations}
\begin{align}
\E{V_t} &= e^{-\lambda t} V_0 + \left(1 - e^{-\lambda t}\right) \kappa \label{eq:V}\\
\E{H_t(u)V_t} &= \left\{\partial_w \phi \big(u, \psi_{T-t}(u,0)\big) + V_0 \partial_w \psi\big(u, \psi_{T-t}(u,0)\big)\right\} e^{uX_0} h(u,t,V_0), \label{eq:HV}\\
\E{H_t(u_1) H_t(u_2) V_t} &= \left\{\partial_w \phi\big(u, q_{T-t}(u_1,u_2)\big) + V_0 \partial_w \psi \big(u, q_{T-t}(u_1,u_2)\big)\right\}  \cdot \label{eq:HHV}\\
&\phantom{=} e^{(u_1 + u_2)X_0} h(u_1,t,V_0) h(u_2,t,V_0),\notag
\end{align}
\end{subequations}
with
\[q_t(u_1,u_2) = \psi_t(u_1,0) +  \psi_t(u_2,0), \qquad h(u,t,V_0) = \exp\Big(\phi_t(u,0) + V_0 \psi_t(u,0)\Big).\]
\end{theorem}
\begin{remark}\label{rem:rho}
Note that the common leading factor $\sigma^2 (1 - \rho^2)$ of $A$, $B$ and $C$ also becomes the leading factor of the minimal squared hedging error $\epsilon^2$, cf.\@ \eqref{eq:outer_ABC}. This makes perfect sense, since it makes the hedging error roughly proportional to vol-of-vol $\sigma$ and shows that the hedging error vanishes in the complete-market boundary cases $\rho = \pm 1$ of the Heston model. However, $\rho$ and $\sigma$ also appear inside $\phi$, $\psi$ and therefore their influence on $\epsilon^2$ is not limited to the leading factor $\sigma^2 (1 - \rho^2)$ alone. \end{remark}

\begin{remark} The domain $D_T$ of finite moments in the Heston model has been described in \cite[Prop.~3.1]{AP07} (see also \cite{FK10}). Using these results, the moment condition in Theorem~\ref{thm:heston} can be checked in the following way: Set $\chi(u) := \rho\sigma u -\lambda$, $\Delta(u) := \chi(u)^2  -\sigma^2(u^2 - u)$ and define
\begin{equation}\label{eq:exp.time}
T_*(u) =\begin{cases}+\infty\,,&\quad \Delta(u)\geq0,\ \chi(u)<0\,;\\\\
\frac{1}{\sqrt{\Delta(u)}}\,\log\Big(\frac{\chi(u)+\sqrt{\Delta(u)}}{\chi(u)-\sqrt{\Delta(u)}}\Big)\,,&\quad \Delta(u)\geq0,\ \chi(u)>0\,;\\\\
\frac{2}{\sqrt{-\Delta(u)}}\,\Big(\arctan\frac{\sqrt{-\Delta(u)}}{\chi(u)}+\pi1_{\{\chi(u)<0\}}\Big)\,,&\quad \Delta(u)<0\,.
\end{cases}
\end{equation}
By \cite[Prop.~3.1]{AP07}, the moment condition $\E{e^{2R_i X_T}} < \infty$ is equivalent to $T < T_*(2R_i)$.
\end{remark}


\section{Numerical results}\label{sec:numerics}
The following numerical implementation should be considered in terms of a `stylized financial market' setting, i.e., while we do not calibrate the model to current market data, we use parameters that are realistic in a market setting. More specifically, we use the Heston model parameters from \cite{gatheral2006volatility}: 
\begin{equation}\label{eq:par}
\begin{aligned}
\kappa &= 0.0354 & \quad \lambda &= 1.3253 &\quad  \rho &= -0.7165 \\ \sigma &= 0.3877 & \quad V_0 &= 0.0174 
\end{aligned}
\end{equation}
In Subsection~\ref{sub:rho} we vary the leverage parameter $\rho$, but keep all other parameters fixed. The current stock price is normalized to $S_0 = 100$ and we use a time-to-maturity of $T=1$ (years) for the variance swap and the call options. The price of a variance swap (i.e.\@ the swap rate $k_* = \E{[\log S, \log S]_T}$) can be readily calculated as
\[k_* = \int_0^T \E{V_t} dt = \kappa T + (V_0 - \kappa) \frac{1 - e^{-\lambda T}}{\lambda} = 0.025427.\]
The supplementary assets are OTM-puts and OTM-calls with strikes ranging from 
\[K_\text{min} = 50 \; \text{ to } \; K_\text{max}  = 150 \text{ in steps of } \; \Delta K = 5.\]

We focus on three aspects of the semi-static hedging problem: 
\begin{itemize}
\item Comparing the different methods that were proposed in Section~\ref{sec:sparse} to solve the sparse semi-static hedging problem;
\item Analyzing the dependency of hedging error and optimal portfolio composition on effective portfolio size $d$;
\item Analyzing the dependency of hedging error and optimal portfolio composition on the leverage parameter $\rho$.
\end{itemize}

\subsection{Comparison of methods}\label{sub:methods}
As a first step, we computed $A$, $B$ and the matrix $C$ from Theorem~\ref{thm:semistatic} using the Fourier-representation in Theorem~\ref{thm:heston} by adaptive integration in MATLAB. 
Next, we implemented the methods, described in Section~\ref{sec:sparse}, i.e.\@ 
\begin{enumerate}[(1)]
\item Greedy forward selection (with and without short-sale constraints)
\item Leaps-and-Bounds (with and without short-sale constraints)
\item LASSO
\end{enumerate}
in \texttt{R}, \cite{rct2016}; using the function \texttt{lars} in the package \texttt{lars} \cite{hastie2013lars} with option \texttt{type="lasso"} for the computation of the LASSO solution. While computationally most demanding, the Leaps-and-Bounds solution can serve as a benchmark solution, since it is (up to numerical error) the exact solution of the sparse semi-static hedging problem \eqref{eq:l0}. The other methods, in contrast, only return a `reasonably close' solution to \eqref{eq:l0}. In all cases, we report the \emph{relative hedging error} 
$\epsilon / k_*$, i.e.\@ the hedging error normalized by the price of the variance swap.

\begin{figure}[htbp] 
  \centering
     \includegraphics[width=\textwidth]{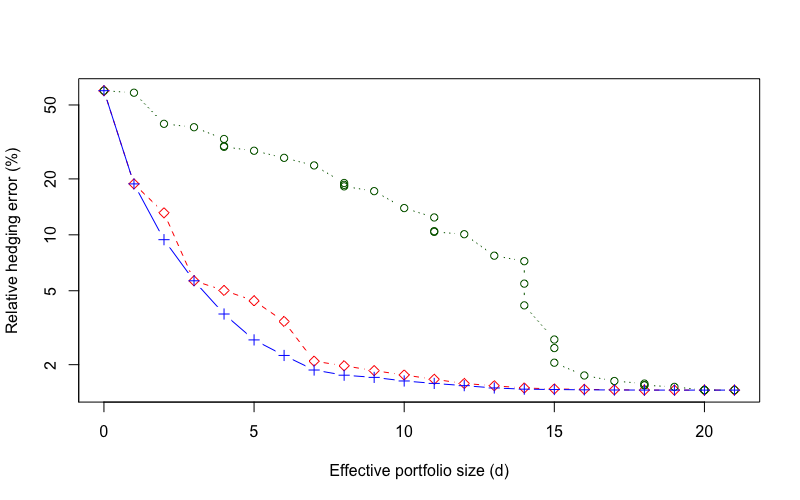}
  \caption{Relative hedging error (on log-scale) for sparse semi-static hedging of a variance swap with different effective portfolio sizes $d$. The plot compares the solutions obtained with the Leaps and Bounds method (blue crosses), greedy method (red diamonds) and LASSO (green circles).}
  \label{fig:error}
\end{figure}

A challenge that is faced by all methods is the bad condition of the matrix $C$. With parameters chosen as above \eqref{eq:par} the reciprocal condition number of $C$ is $\num{1.11e-06}$. While small, this number is still several orders of magnitude larger than the machine precision of $\num{2.22e-16}$ (double precision arithmetic) on the computer that was used. The bad condition of $C$ is not surprising, since put and call options with neighboring strikes are highly correlated. This effect is likely amplified by the fact that $C$ contains the correlations of the GKW-residuals and not the correlations of the option prices themselves. 
While we have considered pre-conditioning of $C$, along the lines of \cite{neumaier1998solving}, we have found that greedy forward selection and Leaps-and-Bounds perform well even without additional conditioning. Also the addition of short-sale constraints seems to have a regularizing effect on the methods.

Figure~\ref{fig:error} shows the relative hedging error (as percentage of the variance swap price) attained with the optimal portfolio returned by methods 1-3 for different effective portfolio sizes $d=0\dots21$. Notice that the implementation of LASSO adds \emph{and removes} supplementary assets from the active set, such that the graph can show multiple solutions for the same effective portfolio size (e.g. for $d=15$). Focusing on the comparison of methods, we find that 

\begin{itemize}
\item The Leaps-and-Bounds method returns the solution with the smallest hedging error, consistent with the fact that it solves \eqref{eq:l0} \emph{exactly}. It is remarkably fast, but  further numerical experiments indicate that its runtime is sensitive to the choice of model parameters. 
\item The greedy method is the fastest method and the residual hedging error of its solution is only slightly higher than the hedging error of the Leaps-and-Bounds solution. Moreover, the performance of the greedy method is stable with respect to parameter choice. 
\item The LASSO methods seems to be severely affected by the bad condition of $C$. This is not surprising, since it has been remarked e.g. in \cite[Sec.~2.6]{buhlmann2011statistics} that the LASSO method has problems with highly correlated data. 
\end{itemize}
Summing up, we can recommend the greedy method as fast, reliable and easy to implement. The Leaps and bounds methods is useful as an efficient way to compute an exact benchmark solution. We cannot recommend LASSO, as it cannot deal well with the bad condition of $C$.\footnote{It should be said, in all fairness, that the \texttt{R}-function \texttt{lars} also provides the option \texttt{"stepwise"} instead of \texttt{"lasso"} which effectively corresponds to the greedy method.} Interestingly, this observations are contrary to the usual wisdom in variable selection for regression problems, where greedy forward selection often has unstable performance and LASSO yields superior results, cf.\@ \cite[Ch.~2]{buhlmann2011statistics}. We attribute these findings to the highly correlated nature of the matrix $C$, which is untypical in regression scenarios, but a natural feature of our hedging problem.

\subsection{Analysis of the hedging error}\label{sub:error}
We return to Figure~\ref{fig:error} to analyze the hedging error resulting from the sparse variance-optimal semi-static hedging problem~\eqref{eq:l0} for different effective portfolio sizes $d$. We consider the benchmark solution returned by the Leaps-and-Bounds method with short-sale constraints. First, we note that dynamic hedging in the underlying $S$, \emph{without} using any static positions in puts and calls ($d=0$) results in a relative hedging error of $\SI{59.7}{\percent}$. This error is already reduced to $\SI{5.7}{\percent}$ by just adding three supplementary assets ($d=3$) and can be further reduced to $\SI{3.4}{\percent}$ by selecting six supplementary assets $(d=6)$. Finally, the error levels off to $\SI{1.6}{\percent}$ when the full range ($d=21$) of puts and calls between $K_\text{min} = 50$ and $K_\text{max} =150$ is used. Further substantial reductions of the hedging error can only be achieved by extending the range of available strikes; adding more options within the current range has only negligible effects. 

The sharp decrease of the hedging error between $d=0$ and $d=3$ affirms the basic premise of \emph{sparse semi-static hedging}: That selecting only a small number of supplementary assets already leads to a significant reduction of the hedging error. On the other hand, the poor performance of the LASSO solution shows that a sub-optimal choice of supplementary assets does not result in a satisfactory reduction of the hedging error. In other words, it is important that the sparse sub-portfolios are chosen optimally, and not arbitrarily.

\subsection{Composition of the hedging portfolios}\label{sub:portfolio}
\begin{sidewaysfigure}
 \centering%
\subfloat[Leaps and Bounds with Short-Sale Constraints]{\includegraphics[width=0.48\textheight]{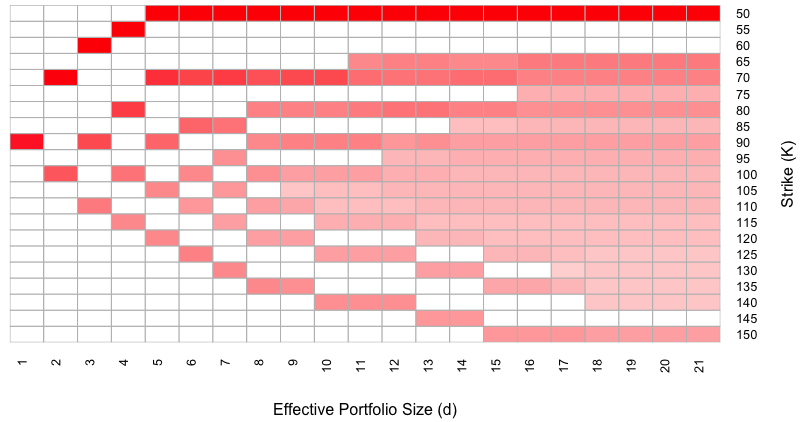}\label{fig:leaps_ssc}}\quad%
\subfloat[Leaps and Bounds without Short-Sale Constraints]{\includegraphics[width=0.48\textheight]{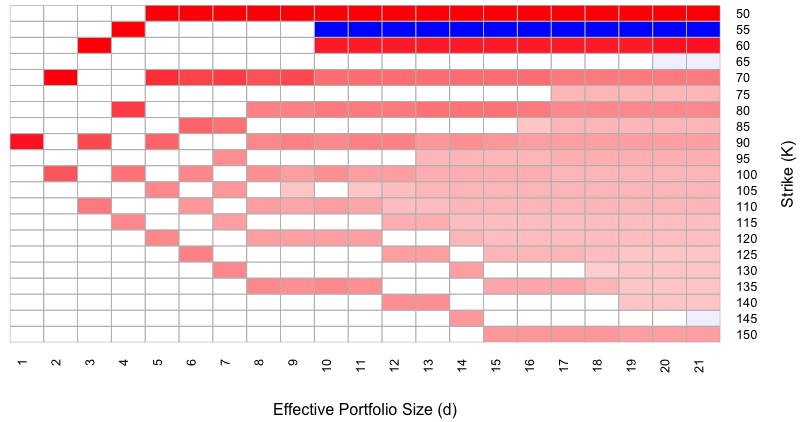}\label{fig:leaps}}\\%
\subfloat[Greedy forward selection with Short-Sale Constraints]{\includegraphics[width=0.48\textheight]{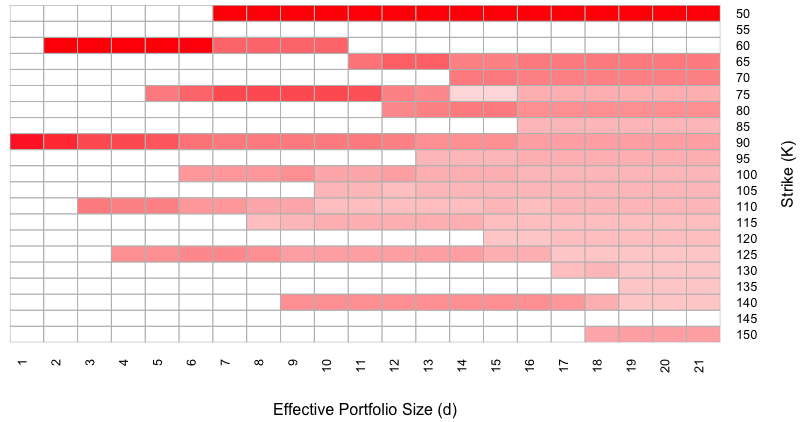}\label{fig:greedy}}\quad%
\subfloat[LASSO]{\includegraphics[width=0.48\textheight]{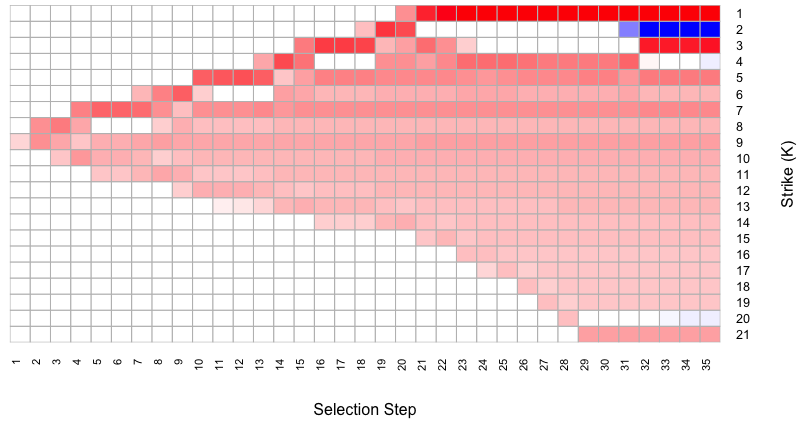}\label{fig:lasso}}%
\caption{Composition of optimal hedging portfolios in dependency on effective portfolio size. Long positions are shown in red and short positions in blue; color saturation corresponds to position size $v(K)$. Different subplots correspond to different solution methods.\label{fig:portfolio}}%
\vspace{-40em}
\end{sidewaysfigure}
We now turn to the composition of the static hedging portfolio, i.e.\@ the vector $v \in \RR^n$ with the constraint $\norm{v}_0 \le d$, that is returned by the solution methods for the sparse semi-static hedging problem~\eqref{eq:l0}. Recall that the element $v_i$ is the nominal size of the position in the supplementary asset $H^i$, with negative sign indicating a short position. In our setting, the elements of $v$ can simply be indexed by the strike $K$ of the corresponding put/call. The optimal portfolios returned by the different solution methods, along with their dependency on effective portfolio size $d$ are shown in Figure~\ref{fig:portfolio}. We make the following observations:

\begin{itemize}
\item With the exception of the put $K = 55$ only long positions are observed;
\item Positions in OTM puts ($K < 100$) are larger than in OTM calls ($K > 100$), in line with Neuberger's replicating portfolio \eqref{eq:swap_replication};
\item The general pattern (going from effective portfolio size $d=1$ to $21$) for all methods can be described as follows: Start with an (approximately) ATM option. Proceed by selecting both OTM puts and calls, going outwards as $d$ increases and putting more weight on OTM puts, until the limit $K_\text{min} = 50$ is reached. Continue by adding OTM calls  and by filling up the gaps from earlier stages. 
\end{itemize}

We suspect that the rare short positions are numerical artifacts, rather than belonging to the true optimal solution of~\eqref{eq:l0}. Indeed, their effect on the hedging error is minuscule, and we hence recommend to use a-priori short-sale constraints, in the case of hedging a variance swap. 

\begin{figure}[htbp] 
  \centering
     \includegraphics[width=\textwidth]{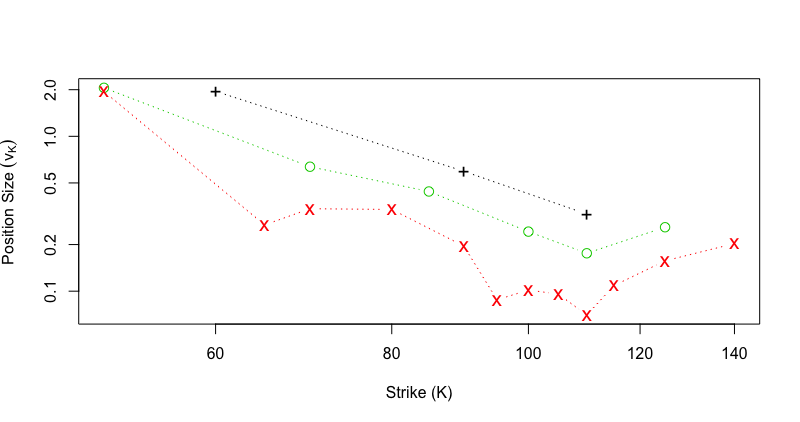}
  \caption{Portfolio weights $v_K$ in the optimal hedging portfolios of effective size $d=3$ (black crosses), $d=6$ (green circles) and $d=12$ (red x's) in doubly logarithmic coordinates. }
  \label{fig:portfolio_weights}
\end{figure}

Figure~\ref{fig:portfolio} gives a good overview of the portfolio composition, but it is difficult to assess the precise size of the individual positions $v_i$. For this reason, we provide in Figure~\ref{fig:portfolio_weights} an additional plot of the portfolio weights $v$ indexed by strike $K$ for the optimal portfolios of effective sizes $d=3,6,12$ in doubly logarithmic coordinates. Note that Neuberger's replicating portfolio \eqref{eq:swap_replication} puts an infinitesimal weight of $v(K)dK = \frac{1}{K^2}dK$ on an option with strike $K$. In doubly logarithmic coordinates, this becomes
\[\log v(K) = - 2 \log K, \]
i.e.\@ in the portfolio weights should form a line of downward slope $-2$. Figure~\ref{fig:portfolio_weights} shows reasonable agreement with this asymptotic result, even for effective portfolio size as small as $d=3$. For $d=12$ numerical errors from the bad condition of the matrix $C$ seem to accumulate and could explain the unruly shape of the graph.

\subsection{The role of correlation}\label{sub:rho}

Finally, we turn to the role of the correlation parameter $\rho$, which is interesting for several reasons: First, the value of $\rho$ does not affect the theoretical price of the variance swap. Second, $\rho$ also does not affect the infinitesimally optimal strategy \eqref{eq:swap_replication}. Finally, $\rho$ allows to tune the degree of market incompleteness, since the Heston model becomes a complete market model in the boundary cases $\rho = \pm 1$. Despite of the first two points, it turns out that $\rho$ has a significant effect on the attainable hedging error and the composition of the optimal portfolio in the sparse semi-static hedging problem. This influence can already be suspected from the leading factor $1 - \rho^2$ appearing in Theorem~\ref{thm:heston}, which propagates to the (squared) hedging error itself, see also Remark~\ref{rem:rho}. Indeed, as Figure~\ref{fig:rho_plot} shows, the dependency of the relative heading error on $\rho$ is very close to a `semi-circle law' $f(\rho) = c_d \sqrt{1 - \rho^2}$, with different constants $c_d$ for different effective portfolio sizes $d$. 

\begin{figure}[htbp] 
  \centering
     \includegraphics[width=\textwidth]{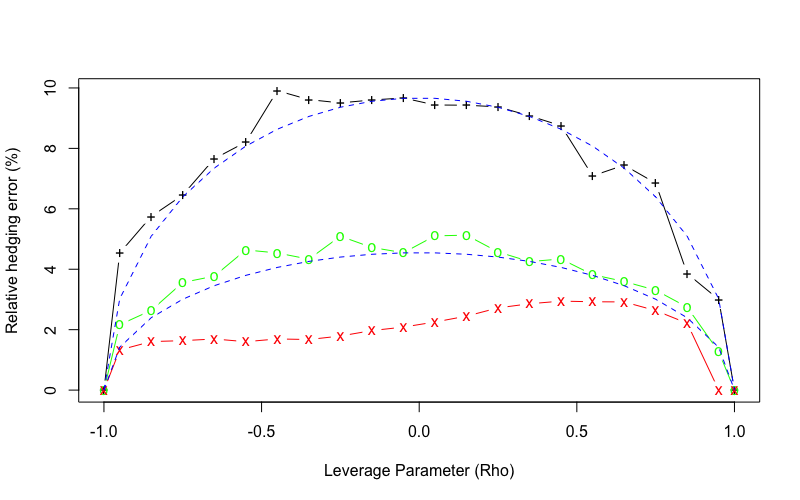}
  \caption{Relative hedging error attainable with a portfolio of effective size $d=3$ (black crosses), $d=6$ (green circles) and $d=12$ (red x's) in relation to the leverage parameter $\rho$. Also shown are the graphs of $f(\rho) = c_d \sqrt{1 - \rho^2}$ (blue dashes) with $c_d$ chosen to fit the red and blue graphs.}
  \label{fig:rho_plot}
\end{figure}

\bibliographystyle{alpha}
\bibliography{bibliography}

\begin{appendix}
\section{The proof of Theorem~\ref{thm:heston}}\label{app:heston}
We show Theorem~\ref{thm:heston} with the help of two lemmas.

\begin{lemma}\label{lem:DT}The set $\cD_T$ (from \eqref{eq:real_moments}) and the function $\psi_t(u,w)$ (from \eqref{eq:sol.psi}) have the following properties:
\begin{enumerate}[(a)]
\item The sets $\cD_T$ are open and convex.
\item If $(u,w) \in S(\cD_T)$ then $(u,\psi_{T-t}(u,w)) \in S(\cD_t)$. 
\item The functions $\phi_t(u,w)$ and $\psi_t(u,w)$ are analytic on $S(\cD_t)$.
\item If $(a,b)  \in \cD_T$, then $(a,b') \in \cD_T$ for all $b' \le b$.
\item $\Re \psi_t(u,w) \le \psi_t(\Re u, \Re w)$ for all $(u,w) \in S(\cD_t)$
\end{enumerate}
\end{lemma}
\begin{proof}
Properties (a), (b) and (c) are shown in \cite{filipovic2009affine}. For (d), note that $V_T \ge 0$ implies that 
\[\E{\exp\left(aX_T + b'V_T\right)} \le \E{\exp\left(aX_T + bV_T\right)}\]
for all $b' \le b$. For (e), note that that Jensen's inequality implies
\begin{multline*}
\exp\left( \Re \phi_T(u,w)+\Re \psi_T(u,w) V_0 +\Re u X_0\right) = \big| \E{\exp\left(uX_T + wV_T\right)}\big| \le \\ \le \E{\big|\exp\left(uX_T + wV_T\right)\big|} 
= \exp\left(\phi_T(\Re u,\Re w)+\psi_T(\Re u,\Re w) V_0 + \Re u X_0\right).
\end{multline*}
Since $V_0$ can be chosen arbitrarily large, (e) follows.
\end{proof}

\begin{lemma}\label{lem:Vdiff} Let $(X,V)$ be given by the Heston model \eqref{eq:Heston} and assume that $(u,w) \in \cS(D_t)$. Then
\[\E{e^{uX_t + wV_t} V_t} = \left\{ \partial_w \phi_t(u,w) + V_0 \partial_w \phi_t(u,w)\right\} e^{X_0} h(u,t,V_0).\]
\end{lemma}
\begin{proof}Fix $(a,b) \in D_t$ and consider $(u,w) \in \cS(D_t)$ of the form $(u = a + \ii y, w = b + \ii z)$. By assumption $K = \E{e^{aX_t + bV_t}}$ exists and is a number in $(0,\infty)$. Define a probability measure $\MM$ on $(\Omega, \cF_t$) by $\left.\frac{\rmd\MM}{\rmd\QQ}\right|_{\cF_t} = \exp(aX_t + bV_t)/K$, i.e.\@ by exponential tilting of $\QQ$. Clearly, the characteristic function of $(X_t,V_t)$ under $\MM$ is given by 
\[\Ex{\MM}{e^{iyX_t + izV_t}} = \Ex{\QQ}{e^{uX_t + wV_t}} = \exp\left(\phi_t(u,w) + V_0 \psi_t(t,u,w) + uX_0\right).\]
Due to the analyticity properties of $\phi_t(u,w)$ and $\psi_t(u,w)$, cf.\@ Lemma~\ref{lem:DT}(c), all partial derivatives of the left hand side with respect to $(y,z)$ exist. Standard results on differentiability of characteristic functions (cf.\@ \cite[Sec.~2.3]{lukacs1960characteristic}) yield that 
\[\Ex{\MM}{e^{iyX_t + izV_t} V_t} = - i \frac{\rmd}{\rmd z} \Ex{\MM}{e^{iyX_t + izV_t}} = \frac{\rmd}{\rmd w} \exp\left(\phi_t(u,w) + V_0 \psi_t(u,w) + uX_0\right).\]
Transforming the left hand side back to $\QQ$  yields the desired result.
\end{proof}

\begin{proof}[Proof of Theorem~\ref{thm:heston}]
First, we determine the relevant quantities of Proposition~\ref{thm:markov} in case of the Heston model. Using \eqref{eq:joi.char} and \eqref{eq:Heston} we obtain
\begin{align*}
h(u,T-t,V_t) &= \exp\Big(\phi_{T-t}(u,0) + V_t \psi_{T-t}(u,0)\Big), & \partial_v h(u,T-t,V_t) &= \psi_{T-t}(u) h(u,T-t,V_t),\\
\partial_v \gamma(T-t,V_t) &= \frac{1}{\lambda} \left(1 - e^{-\lambda (T-t)}\right), & dQ &= \sigma^2 (1 - \rho^2) V_t\,dt.
\end{align*}
Therefore, by Proposition~\ref{thm:markov},
\begin{align*}
\rmd\cA &= \frac{\sigma^2 (1 - \rho^2)}{\lambda^2}  \left(1 - e^{-\lambda(T-t)}\right)^2 V_t \,dt\\
\rmd\cB(u) &= \frac{\sigma^2 (1 - \rho^2)}{\lambda} \left(1 - e^{-\lambda(T-t)}\right) \psi_{T-t}(u) H_t(u)) V_t\,dt\\
\rmd\cC(u_1,u_2) &= \sigma^2 (1 - \rho^2) \psi_{T-t}(u_1) \psi_{T-t}(u_2) H_t(u_1) H_t(u_2)  V_t\,dt.
\end{align*}
If the expectations in \eqref{eq:expectations} are finite, then an application of Theorem~\ref{thm:markov} yields the desired representations of $A,B, C$. Thus, it remains to show integrability and to determine the explicit expressions in \eqref{eq:expectations}.\\
First, \eqref{eq:V} is easily obtained from the Heston SDE \eqref{eq:Heston}. To show \eqref{eq:HV} we make use of Lemma~\ref{lem:DT} and~\ref{lem:Vdiff}. Let $u = x + iz$ be element of some strip $\cS(R_j)$ and note that the integrability condition on $X_T$ implies that $(x,0) \in \cD_T$. From Lemma~\ref{lem:DT}(b) we conclude that $(x, \psi_{T-t}(x,0)) \in \cD_t$. Now $\Re \psi_{T-t}(u,0) \le \psi_{T-t}(0,x)$, together with Lemma~\ref{lem:DT}(d) shows that also $(\Re  u, \Re \psi_{T-t}(0,u)) \in \cD_t$, which is equivalent to $(u, \psi_{T-t}(u,0)) \in S(\cD_t)$. Applying Lemma~\ref{lem:Vdiff} with $w = \psi_{T-t}(u,0)$ yields \ref{eq:HV}.\\
For \eqref{eq:HHV} we can use a similar argument: Write $u_1 = x_1 + i z_1$ and $u_2 = x_2 + i z_2$. The integrability condition on $X_T$ implies that $(2x_1,0)$ and $(2x_2,0)$ are in $\cD_T$. From Lemma~\ref{lem:DT}(b) we conclude that $(2x_1,\psi_{T-t}(2x_1,0)) \in \cD_t$, and similarly for $x_2$. Convexity of $\cD_t$, see Lemma~\ref{lem:DT}(a), shows that $(x_1 + x_2, \frac{1}{2}q_{T-t}(2x_1, 2x_2)) \in \cD_T$. Now convexity of $\psi_{T-t}$ (and hence of $q_{T-t}$), together with Lemma~\ref{lem:DT}(d) yields that also $(x_1 + x_2, q_{T-t}(x_1,x_2) \in \cD_t$. To pass to complex arguments, note that Lemma~\ref{lem:DT}(e) implies that also 
$(u_1 + u_2, q_{T-t}(u_1,u_2)) \in S(\cD_T)$. Hence, we may apply Lemma~\ref{lem:Vdiff} with $u = u_1 + u_2$ and $w = q_{T-t}(u_1, u_2)$, which yields \eqref{eq:HHV}.
\end{proof}
\end{appendix}

\end{document}